\documentclass[long]{llncs}
\usepackage{amssymb,amsmath,}
\usepackage{comment}
\usepackage{temporal}
\usepackage{gastex}
\usepackage{xspace}
\usepackage{color}
\usepackage[normalem]{ulem} 

\newcommand{\pair}[1]{\tuple{#1}}

\def\fsize{\footnotesize}
\def\loop{6}
\def\Nw{7}
\def\Nh{7}
\def\Nmr{7}

\makeatletter
\renewcommand\paragraph{\@startsection{paragraph}{4}{\z@}%
                       {-3\p@ \@plus -4\p@ \@minus -4\p@}%
                       {-0.5em \@plus -0.22em \@minus -0.1em}%
                       {\normalsize\bfseries}}
\renewcommand\section{\@startsection{section}{1}{\z@}%
                       {-8\p@ \@plus -4\p@ \@minus -4\p@}%
                       {4\p@ \@plus 4\p@ \@minus 4\p@}%
                       {\normalfont\large\bfseries\boldmath
                        \rightskip=\z@ \@plus 8em\pretolerance=10000
                        }}
\renewcommand\subsection{\@startsection{subsection}{2}{\z@}%
                       {-8\p@ \@plus -2\p@ \@minus -2\p@}%
                       {6\p@ \@plus 2\p@ \@minus 2\p@}%
                       {\normalfont\normalsize\bfseries\boldmath
                        \rightskip=\z@ \@plus 8em\pretolerance=10000
                      }}
\makeatother

\spnewtheorem{Example}[example]{Example}{\bf}{\vspace{-1mm}}
\spnewtheorem{Lemma}[lemma]{Lemma}{\bf}{\vspace{-2mm}}

\definecolor{grey}{rgb}{0.5,0.5,0.5}

\renewcommand{\epsilon}{\varepsilon}
\newcommand{\Eproj}{\bar{E}}
\newcommand{\lang}[1]{{L_{#1}}} 
\newcommand{\qual}{L}  
\newcommand{\quan}{L}  
\newcommand{\minimum}{\vec{-1}}

\newcommand{\I}{{\sf I}}
\renewcommand{\O}{{\sf O}}
\newcommand{\AP}{{\sf AP}}

\newcommand{\maps}{\rightarrow}
\renewcommand{\implies}{\maps} 
\newcommand{\real}{\mathbb{R}}

\newcommand{\nat}{\mathbb{N}}

\newcommand{\tuple}[1]{\langle#1\rangle}
\newcommand{\set}[1]{\{#1\}}
\newcommand{\vc}{\vec{c}}
\newcommand{\vctr}[1]{\tuple{#1}}
\def\init{s_0}
\def\vinit{v_0}

\def\strat{\pi}
\def\Strat{\Pi}

\def\p{p} 

\def\Value{{\cal V}}

\def\reward{r}
\def\lreward{{\vec{r}}}
\def\priority{p}
\def\maxpriority{{|\priority|}}
\def\maps{\rightarrow}
\def\game{{\cal G}}
\def\cI{{\cal I}}
\def\cO{{\cal O}}
\def\A{{\Sigma}} 
\def\IA{{\A_I}}
\def\OA{{\A_O}}
\def\M{M}
\def\dim{d}

\def\grant{g}
\def\req{r}
\def\g{g}
\def\r{r}
\def\ng{{\bar{g}}}
\def\nr{{\bar{r}}}

\newcommand{\pat}{\rho} 
\newcommand{\Pat}{\Omega} 
\newcommand{\LM}{\mathit{LM}}
\newcommand{\MP}{\mathit{MP}}
\newcommand{\LMP}{\mathit{LMP}}
\def\ve{\varepsilon}
\renewcommand{\phi}{\varphi}

\newcommand{\wh}{\widehat}

\title{\large Better Quality in Synthesis through Quantitative
  Objectives
}
\author{Roderick Bloem\inst{1} \and Krishnendu Chatterjee\inst{2} \and\\
  Thomas A. Henzinger\inst{3} \and Barbara Jobstmann\inst{3}}
\institute{${}^1$Graz University of Technology,\ \ ${}^2$IST, Austria, \ \ ${}^3$EPFL}

\begin{document}

\maketitle
\vspace{-15pt}
\begin{abstract}
  Most specification languages express only qualitative constraints.
  However, among two implementations that satisfy a given
  specification, one may be preferred to another.  For example, if a
  specification asks that every request is followed by a response, one
  may prefer an implementation that generates responses quickly but
  does not generate unnecessary responses.  We use quantitative
  properties to measure the ``goodness'' of an implementation.  Using
  games with corresponding quantitative objectives, we can
  synthesize ``optimal'' implementations, which are preferred among
  the set of possible implementations that satisfy a given
  specification.

  In particular, we show how automata with lexicographic mean-payoff
  conditions can be used to express many interesting quantitative
  properties for reactive systems.  In this framework, the synthesis
  of optimal implementations requires the solution of lexicographic
  mean-payoff games (for safety requirements), and the solution of
  games with both lexicographic mean-payoff and parity objectives (for
  liveness requirements).  We present algorithms for
  solving both kinds of novel graph games.
\end{abstract}


\section{Introduction}
\label{sec:introduction}

Traditional specifications are Boolean: 
an implementation satisfies a specification, or it does not.
This Manichean view is not entirely satisfactory:
There are usually many different ways to satisfy a specification, 
and we may prefer one implementation over another.
This is especially important when we automatically synthesize
implementations from a specification, because we have no other way to
enforce these preferences.
In this paper, we add a quantitative aspect to system specification,
imposing a preference order on the implementations that 
satisfy the qualitative part of the specification.  
Then, we present synthesis algorithms that construct, from a given 
specification with both qualitative and quantitative aspects, an 
implementation that 
(i)~satisfies the qualitative aspect and 
(ii)~is optimal or near-optimal with respect to the quantitative aspect.
Along the way, we introduce and solve graph games with new kinds of 
objectives, namely, lexicographic mean-payoff objectives and the
combination of parity and lexicographic mean-payoff objectives.

Suppose we want to specify an arbiter for a shared resource.  For each
client $i$, the arbiter has an input $r_i$ (request access) and an
output $g_i$ (access granted).  A first attempt at a specification in
LTL may be $\bigwedge_i \always(r_i \rightarrow \eventually g_i) \,
\wedge\, \always \bigwedge_i \bigwedge_{j \neq i} (\neg g_i \vee \neg
g_j)$.  (All requests are granted eventually and two grants never
occur simultaneously.)  This specification is too weak: An
implementation that raises all $g_i$ signals in a round-robin fashion
satisfies the specification but is probably undesired.  The
unwanted behaviors can be ruled out by adding the requirements
$\bigwedge_i \always(g_i \rightarrow \nextt(\neg g_i \weakuntil
r_i))\, \wedge\, \bigwedge_i \neg g_i \weakuntil r_i$. (No second
grant before a request.)

Such Boolean requirements to rule out trivial but undesirable 
implementations have several drawbacks:
(i)~they are easy to forget and difficult to get right 
(often leading to unrealizable specifications) and, perhaps more 
importantly, 
(ii)~they constrain implementations unnecessarily, by giving up the 
abstract quality of a clean specification.
In our example, we would rather say that the implementation should 
produce ``as few unnecessary grants as possible''
(where a grant $g_i$ is unnecessary if there is no outstanding 
request~$r_i$).
We will add a quantitative aspect to specifications which allows us 
to say that.
Specifically, we will assign a real-valued reward to each behavior, and 
the more unnecessary grants, the lower the reward.

A second reason that the arbiter specification may give rise to
undesirable implementations is that it may wait arbitrarily long
before producing a grant.  Requiring that grants come within a fixed
number of steps instead of ``eventually'' is not robust, because it
depends on the step size of the implementation and the number of
clients.  Rather, we assign a lower reward to executions with larger
distances between a request and corresponding grant.  If we use
rewards both for punishing unnecessary grants and for punishing late
grants, then these two rewards need to be combined.  This leads us to
consider tuples of costs that are ordered lexicographically.  We
define the quantitative aspect of a specification using
\emph{lexicographic mean-payoff automata}, which assign a tuple of
costs to each transition.  The cost of an infinite run is obtained by
taking, for each component of the tuple, the long-run average of all
transition costs.  Such automata can be used to specify both ``produce
as few unnecessary grants as possible'' and ``produce grants as
quickly as possible,'' and combinations thereof.

If the qualitative aspect of the specification is a safety property,
then synthesis requires the solution of \emph{lexicographic
  mean-payoff games}, for which we can synthesize optimal solutions.
(The objective is to minimize the cost of an infinite run lexicographically.)
If the qualitative aspect is a liveness property, then we obtain
\emph{lexicographic mean-payoff parity games}, which must additionally
satisfy a parity objective.  We 
present the solution of these games in this paper.  We
show that lexicographic mean-payoff games are determined for
memoryless strategies and can be decided in NP $\cap$ coNP, but that
in general optimal strategies for lexicographic mean-payoff parity
games require infinite memory.  
We
prove, however, that for any given real vector
$\vec{\varepsilon}>\vec{0}$, there exists a finite-state strategy that
ensures a value within $\vec{\varepsilon}$ of the optimal value.  This
allows us to synthesize $\varepsilon$-optimal implementations, for any
$\vec{\varepsilon}$.  
The complexity class of the optimal synthesis problem is NP.

\smallskip\noindent{\em Related work.}  There are several formalisms
for quantitative specifications in the
literature~\cite{Alur09,CCHKM05,CAHS03,CAFHMS06,Chatte08a,dA98,DiscountingTheFuture,DrosteGastin07,DrosteKR08,KL07};
most of these works (other than~\cite{Alur09,Chatte08a,dA98}) do not
consider mean-payoff specifications and none of these works focus on
how quantitative specifications can be used to obtain better
implementations for the synthesis problem.
Several notions of metrics
have been proposed in the literature for probabilistic systems and
games~\cite{dAMRS07,DGJP99}; these metrics provide a measure that
indicates how close are two systems with respect to all temporal
properties expressible in a logic; whereas our work compares how good
an implementation is with respect to a given specification.  The
work~\cite{CHJ04} considers non-zero-sum games with lexicographic
ordering on the payoff profiles, but to the best of our knowledge, the
lexicographic quantitative objective we consider for games has not
been studied before.




\section{Examples}
\label{sec:qualityToQuantity}

After giving necessary definitions, we illustrate with several
examples how quantitative constraints can be a useful measure for the
quality of an implementation.

\label{sec:definitions}

\paragraph{Alphabets, vectors, and lexicographic order.}
Let $\cI$ and $\cO$ be finite sets of \emph{input} and \emph{output
  signals}, respectively.  We define the \emph{input alphabet} $\IA =
2^{\cI}$ and the \emph{output alphabet} $\OA=2^{\cO}$. The joint
alphabet $\A$ is defined as $\A = 2^{\cI \cup \cO}$.
Let $\real^d$ be the set of real vectors of dimension $\dim$ with the
usual lexicographic order.  

\paragraph{Mealy machines.}
A \emph{Mealy machine} is a tuple $M = \tuple{Q,q_0,\delta}$, where
$Q$ is a finite set
of states, $q_0 \in Q$ is
the initial state, and $\delta \subseteq Q \times \IA \times \OA \times Q$ is a
set of labeled edges.  We require that the machine is \emph{input enabled} and
\emph{deterministic}:
$\forall q \in Q\scope \forall i \in \IA$, there exists a
unique $ o \in \OA$ and a unique $ q' \in
Q$ such that $\tuple{q,i,o,q'} \in \delta$.
Each input word $i = i_0 i_1 \dots \in \IA^{\omega}$ has a unique
\emph{run} $q_0 i_0 o_0 q_1 i_1 o_1 \dots$ such that $\forall k\ge 0
\scope \tuple{q_k, i_k, o_k, q_{k+1}} \in \delta$.  The corresponding
\emph{I/O word} is $i_0 \cup o_0, i_1 \cup o_1,\dots \in \A^{\omega}$.
The \emph{language} of $M$, denoted by $\lang{M}$, is the set of all
I/O words of the machine.  Given a language $\qual \subseteq
\Sigma^{\omega}$, we say a Mealy machine $M$ \emph{implements} $\qual$
if $\lang{M}\subseteq \qual$.

\paragraph{Quantitative languages.}  A \emph{quantitative language
  \cite{Chatte08a} over $\Sigma$} is a function $\quan: \A^\omega
\maps V$ that associates to each word in $\A^\omega$ a \emph{value}
from~$V$, where $V\subset \real^d$ has a
least element.
Words with a higher value are more desirable than those with a lower
value.  In the remainder, we view an ordinary, \emph{qualitative}
language as a quantitative language that maps words in $L$ to $\true$
($=1$) and words not in $L$ to $\false$ (= 0).
We often use a pair $\pair{\qual,\quan'}$ of a qualitative language
$\qual$ and a quantitative language $\quan':\A^\omega \maps V$ as
specification, where $\qual$ has higher
priority than $\quan'$.  We can also view $\pair{\qual,\quan'}$ as
quantitative language with $\pair{\qual,\quan'}(w)= \vec{0}$ if
$\qual(w) = 0$, and
$\pair{\qual,\quan'}(w)=\quan'(w)-v_{\bot}+\vec{1}$ otherwise, where
$v_{\bot}$ is the minimal value in $V$.  (Adding constant factors does
not change the order between words).

We  extend the definition of value to Mealy machines.  As
in verification and synthesis of qualitative languages, we take the
worst-case behavior of the Mealy machine as a measure.
Given a quantitative language $\quan$ over $\A$, the \emph{value} of a
Mealy machine $M$, denoted by $\quan(M)$, is $\inf_{w \in \lang{M}}
\quan(w)$.

\paragraph{Lexicographic mean-payoff automata.}
We use lexicographic mean-payoff automata to describe quantitative
languages.  In lexicographic mean-payoff automata each edge is mapped
to a reward.  The automaton associates a run with a word and assigns
to the word the average reward of the edges taken (as in mean-payoff
games \cite{Ehrenf79}).  Unlike in mean-payoff games, rewards are
vectors.

Formally, a \emph{lexicographic mean-payoff automaton} of dimension
$\dim$ over $\A$ is a tuple $A=\tuple{\tuple{S, \init, E}, \lreward}$,
where $S$ is a set of states, $E \subseteq S \times \Sigma \times S$
is a labeled set of edges, $\init\in S$ is the initial state, and
$\lreward: E \maps \nat^\dim$ is a reward function that maps edges to
$\dim$-vectors of natural numbers.  Note that all rewards are
non-negative.  We assume that the automaton is complete and deterministic:
for each $s$ and $\sigma$ there is exactly one $s'$ such that
$\tuple{s,\sigma,s'}\in E$.
A word $w = w_0 w_1 \dots \in \A^{\omega}$ has a unique run $\rho(w) =
s_0 e_0 s_1 e_1 \dots$ such that $s_i \in S$ and
$e_i=\tuple{s_i,w_i,s_{i+1}} \in E$  for all $i\ge 0$.
The \emph{lexicographic mean
  payoff} $\LM(\rho)$ of a run $\rho$ is defined as $LM(\rho)=\lim\inf_{n
  \rightarrow \infty} \frac{1}{n}\sum_{i=0}^{n} \lreward(e_i)$.
The automaton defines a quantitative language with domain $\real^\dim$
by associating to every word $w$ the value $\quan_A(w) =
\LM(\rho(w))$.

If the dimension of $A$ is 1 and the range of $\quan_A$ is $\{0,1\}$
then, per definition, $\lang{A}$ defines a qualitative language.  We
say that $A$ is a \emph{safety automaton} if it defines a qualitative
language and there is no path from an edge with reward~$0$ to an edge
with reward $>0$.  Safety automata define safety languages
\cite{Alpern85}.  Note that in general, $\omega$-regular languages and
languages expressible with mean-payoff automata are incomparable
\cite{Chatte08a}.

\label{sec:examples}
\begin{Example}\label{ex:intro}
  Let us consider a specification of an arbiter with one client.  In
  the following, we use $\r$, $\nr$, $\g$, and $\ng$ to represent that
  $\req$ or $\grant$ are set to $\true$ and $\false$, respectively and
  $\top$  to indicate that a signal can take either value.  A slash
  separates~input~and~output.

  Take the specification $\phi=\always(\req \implies \grant \vee
  \nextt\grant)$: every request is granted within two steps.  The
  corresponding language~$\lang{\phi}$ maps a word $w=w_0w_1,\dots$ to
  $\true$ iff for every position $i$ in $w$, if $\req\in w_i$, then
  $\grant \in w_{i} \cup w_{i+1}$.
  \begin{figure}[tb]
    \begin{minipage}{0.5\textwidth}
    \centering
      \def\links{-5}
\def\rechts{65}
\def\oben{33}
\def\unten{5}

\begin{picture}(\rechts,\oben)(\links,\unten)

{\fsize
\gasset{loopdiam=\loop,Nw=\Nw,Nh=\Nh,Nmr=\Nmr}

\node[Nmarks=i](q0)(0,15){}  
\node[Nmarks=i](q1)(20,15){} 

\node[Nmarks=i](q3)(40,15){} 
\node[Nmarks=n](q4)(60,15){}  

\put(-3,5){$\M_1$}
\put(17,5){$\M_2$}
\put(37,5){$\M_3$}

\drawloop[ELside=l, ELdist=1](q0){$\top/g$}

\drawloop[ELside=l, ELdist=1](q1){$\begin{array}{c}\nr/\ng\\r/g \end{array}$}



\drawloop[ELside=l, ELdist=1](q3){$\nr/\ng$}
\drawedge[ELpos=50, ELside=l, ELdist=1, curvedepth=6](q3,q4){$\r/\ng$}
\drawedge[ELpos=50, ELside=l, ELdist=1, curvedepth=6](q4,q3){$\top/\g$}

}
\end{picture}

      \caption{Three Mealy machines that implement $\always(\req \implies
        \grant\vee \nextt\grant)$}
    \label{fig:systems}
    \end{minipage}
    \hspace{0.04\textwidth}
    \begin{minipage}{0.45\textwidth}
      \centerline{\def\rechts{75}
\def\oben{33}
\def\links{-5}
\def\unten{5}

\begin{picture}(\rechts,\oben)(\links,\unten)
{\fsize
\put(0,20){$A_1$}
\gasset{loopdiam=\loop,Nw=\Nw,Nh=\Nh,Nmr=\Nmr}

\node[Nmarks=i](q0)(10,15){$q_0$}
\drawloop[ELside=l, ELdist=1](q0){$\ng(1)$}
\drawloop[ELside=l, ELdist=1, loopangle=270](q0){$\g(0)$}

\put(25,20){$A_2$}
\node[Nmarks=i](q0)(35,15){$q_0$}
\node[Nmarks=n](q1)(55,15){$q_1$}

\drawloop[ELside=l, ELdist=1](q0){$\begin{array}{c}\nr\ng(1)\\ \r\g(1)\end{array}$}
\drawloop[ELside=l, ELdist=1, loopangle=270](q0){$\nr\g(0)$}
\drawloop[ELside=l, ELdist=1](q1){$\ng(1)$}

\drawedge[ELpos=50, ELside=l, ELdist=1, curvedepth=6](q0,q1){$\r\ng(1)$}
\drawedge[ELpos=50, ELside=l, ELdist=1, curvedepth=6](q1,q0){$\g(1)$}

}

\end{picture}

}
      \caption{Two specifications that provide different ways of charging for
        grants.}
      \label{fig:unnecessary}
    \end{minipage}
  \end{figure}
  Fig.~\ref{fig:systems} shows three implementations for
  $\lang{\phi}$.  Machine~$\M_1$ asserts~$\grant$ continuously
  independent of $\req$, $\M_2$ responds to each request with a grant
  but keeps~$\grant$ low otherwise, and $\M_3$ delays its response if
  possible.

  We  use a quantitative specification to state that we prefer an
  implementation that avoids unnecessary grants.
  Fig.~\ref{fig:unnecessary} shows two mean-payoff automata, $A_1$
  and $A_2$ that define rewards for the behavior of an
  implementation.
  Note that we have summarized edges using Boolean algebra.
  For instance, an arc labeled $\g$ in the figure corresponds to the edges
  labeled $\r\g$ and $\nr\g$.
  Automata~$A_1$ and $A_2$ define quantitative languages that
  distinguish words by the frequency of grants and the condition under
  which they appear.  Specification~$A_1$ defines a reward of~1 except
  when a grant is given; $A_2$ only withholds a reward when a grant is
  given unnecessarily.
  Consider the  words $w_1=(\r\g, \nr\ng)^\omega$ and
  $w_2=(\r\ng, \nr\g, \nr\g)^\omega$.  Specification~$A_1$ defines the
  rewards $\quan_{A_1}(w_1) = 1/2$, 
  and $\quan_{A_1}(w_2)=1/3$.  For $A_2$, we get
  $\quan_{A_2}(w_1)=1$ 
  and  $\quan_{A_2}(w_2)=2/3$.
  Both specifications are meaningful but they express different
  preferences, which leads to different results for verification and
  synthesis, as discussed in Section~\ref{sec:verification-synthesis}.


  Recall the three implementations in Fig.~\ref{fig:systems}.  Each
  of them implements $\lang{\phi}$.
  For $A_1$, input $\r^{\omega}$ gives the lowest reward.  The values
  corresponding to the input/output word of $\M_1$, $\M_2$, and $\M_3$
  are $0$, $0$, and $1/2$, respectively.  
  Thus, $A_1$ prefers the last implementation.  The 
  values of the implementations for $A_2$ are minimal when the input
  is $\nr^{\omega}$; they are $0$, $1$, and $1$, respectively. 
  Thus, $A_2$ prefers the last two
  implementations, but does not distinguish between them.
\end{Example}

\begin{Example}\label{ex:tuple}
  Assume we want to
  specify an arbiter for two clients that answers requests within
  three steps.  Simultaneous grants
  are forbidden. Formally, we have
  $\phi = \bigwedge_{i\in\set{1,2}} \always\bigl(\req_i \implies
  \bigvee_{t\in\set{0,1,2}} \nextt^t\grant_i\bigr) \wedge
  \always( \neg g_1 \vee \neg g_2).$
  We want grants to come as quickly as possible.
  Fig.~\ref{fig:distance} shows a mean-payoff automaton $A_3$ that
  rewards rapid replies to Client~1.  Suppose we want to do the same
  for Client~2.  One option is to construct a similar automaton~$A'_3$
  for Client~2 and to add the two resulting quantitative languages.
  This results in a quantitative language $\quan_{A_3} + \quan_{A'_3}$ that treats
  the clients equally.  Suppose instead that we want to give
  Client~1 priority.  In that case, we can construct a lexicographic
  mean-payoff automaton that maps a word $w$ to a tuple
  $\tuple{s_1(w),s_2(w)}$, where the first and second elements
  correspond to the payoff for Clients~1 and~2, resp.  Part of this
  automaton, $A_{4}$, is shown in  Fig.~\ref{fig:distance}.
  
  \begin{figure}[tb]
    \centerline{\def\rechts{85}
\def\oben{45}
\def\links{-12}
\def\unten{18}

\begin{picture}(\rechts,\oben)(\links,\unten)
{\fsize
\gasset{loopdiam=\loop,Nw=\Nw,Nh=\Nh,Nmr=\Nmr}

\put(-10,15){$A_3$}
\node[Nmarks=i](q0)(0,35) {$q_0$}
\node[Nmarks=n](q1)(0,20){$q_1$}

\drawloop[ELside=l, ELdist=1,loopangle=90](q0){
  $\begin{array}{c}
    \nr_1(1)\\
    \g_1(1)
    \end{array}$}
\drawloop[ELside=l, ELdist=1,loopangle=0](q1){$\ng_1(0)$}

\drawedge[ELpos=50, ELside=r, ELdist=1, curvedepth=0](q0,q1){$\r_1\ng_1(0)$}
\drawedge[ELpos=50, ELside=r, ELdist=1, curvedepth=-3](q1,q0){$\g_1(1)$}

\put(40,15){$A_{4}$}
\node[Nmarks=i](q10)(50,35) {$q'_{00}$}
\node[Nmarks=n, dash={1}0](q11)(50,20){$q'_{10}$}
\node[Nmarks=n, dash={1}0](q12)(80,35){$q'_{01}$}
\node[Nmarks=n, dash={1}0](q13)(80,20){$q'_{11}$}

\drawloop[ELside=l,ELdist=1,loopangle=90](q10){
  $\begin{array}{c}
    \nr_1\nr_2 \begin{tiny}\vctr{1,1}\end{tiny}\\
    \nr_1\g_2 \begin{tiny}\vctr{1,1}\end{tiny}\\
    \g_1\nr_2 \begin{tiny}\vctr{1,1}\end{tiny}\\
    \g_1\g_2 \begin{tiny}\vctr{1,1}\end{tiny}\\
  \end{array}$}

\drawedge[ELpos=50, ELside=r, ELdist=1, curvedepth=0](q10,q11){
  $\begin{array}{c}
    \r_1\ng_1\g_2 \begin{tiny}\vctr{0,1}\end{tiny}\\
    \r_1\ng_1\nr_2 \begin{tiny}\vctr{0,1}\end{tiny}\end{array}$}
\drawedge[ELpos=50, ELside=l, ELdist=1, curvedepth=0](q10,q12){
  $\begin{array}{c}
    \g_1\r_2\ng_2 \begin{tiny}\vctr{1,0}\end{tiny}\\
    \nr_1\r_2\ng_2 \begin{tiny}\vctr{1,0}\end{tiny}\end{array}$}
\drawedge[ELpos=50, ELside=l, ELdist=1, curvedepth=0](q10,q13){$\dots$}



}
\end{picture}

}
    \caption{A specification that rewards quick grants for a request
      from Client~1, and a specification that rewards quick grants
      for both clients, while giving priority to Client~1.}
    \label{fig:distance}
  \end{figure}
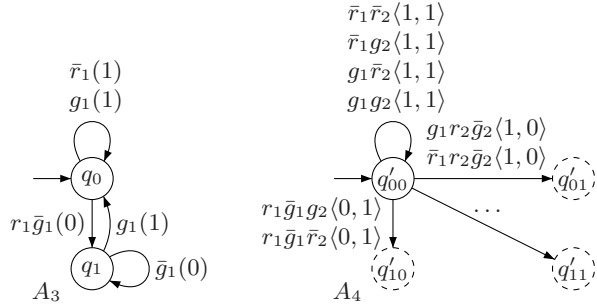
  Automaton $A_3$ distinguishes words with respect to the maximal
  average distance between request and grant.  For instance,
  $\quan_{A_3}((\r_1\g_1, \nr_1\ng_1)^\omega) = 1$ and $\quan_{A_3}((\r_1\ng_1,
  \nr_1\g_1)^\omega) = 1/2$.
%
  Automaton~$A_4$ associates a vector to every word.  For
  instance, 
  $\quan_{A_4}((\r_1\g_1\r_2\ng_2, \nr_1\ng_1\nr_2\g_2)^{\omega}) = 
  1/2\cdot(\vctr{1,0} + \vctr{1,1}) = \vctr{1,1/2}$, which makes
  it preferable to the word $(\r_1\ng_1\r_2\g_2,
  \nr_1\g_1\nr_2\ng_2)^{\omega}$, which has value $\vctr{1/2,1}$.
  This is what we expect: the first word gives priority to requests
  from Client~1, while the second gives priority to Client~2.
\end{Example}
\sloppy
\begin{Example}\label{ex:limitrealizable}
  Let us consider the liveness specification $\phi = \always(\r
  \rightarrow \eventually \g)$ saying that every request must be
  granted eventually.  
  This languages can usefully be combined with $A_3$, stating that
  grants must come quickly.  It can also be combined with $A_1$ from
  Fig.~\ref{fig:unnecessary} stating that grants should occur as
  infrequently as possible.  A Mealy machine may emit a grant every
  $k$ ticks, which gives a reward of $1 - 1/k$.  Thus, there is an
  infinite chain of ever-better machines.  There is no Mealy machine,
  however, that obtains the limit reward of 1.  This can only be
  achieved by an implementation with infinite memory, for instance one
  that answers requests only in cycle~$2^k$ for
  all~$k$~\cite{Chatte05}.
\end{Example}

\section{Lexicographic Mean-Payoff (Parity) Games}
\label{sec:games}

We show how to solve lexicographic mean-payoff games and lexicographic
mean-payoff parity games, which we will need in
Section~\ref{sec:verification-synthesis} to solve the synthesis
problem for quantitative specifications.

\subsection{Notation and known results}
\label{sec:preliminaries}


\paragraph{Game graphs and plays.}
A \emph{game graph} over the alphabet $\A$ is a tuple $G=\tuple{S,
  \init, E}$ consisting of a finite set of states $S$, partitioned into
$S_1$ and $S_2$, representing the states of Player~1 and
Player~2, an initial state $\init \in S$, and a finite set
of labeled edges $E \subseteq S \times \Sigma \times S$.  We require
that the labeling is deterministic, i.e., if $\tuple{s,\sigma,t},
\tuple{s,\sigma,t'} \in E$, then $t=t'$.  We write $\Eproj =
\{\tuple{s,t} \mid \exists \sigma \in \Sigma: \tuple{s, \sigma, t} \in
E\}$.  At $S_1$
states, Player~1 decides the successor state and at $S_2$ states,
Player~2 decides the successor states.  We assume that
$\forall s \in S\scope \exists t \in S \scope \tuple{s,t} \in \Eproj$.
A \emph{play} $\pat= \pat_0 \pat_1 \dots \in S^\omega$ is an infinite
sequence of states such that for all $i\ge 0$ we have
$\tuple{\pat_i, \pat_{i+1}} \in \Eproj$.  We denote the set of all
plays by $\Pat$.


The labels and the initial state are not relevant in this section.
They are used later to establish the connection between
specifications, games, and Mealy machines.  They also allow us to view
automata as games with a single player.


\paragraph{Strategies.}
Given a game graph $G=\tuple{S, \init, E}$, a \emph{strategy for Player~1} is
a function $\strat_1: S^*S_1\maps S$ 
such that
$\forall s_0\dots s_i \in S^*S_1$ we have $\tuple{s_i,\strat_1(s_0s_1\dots
  s_i)}\in \Eproj$.  
A Player-2 strategy is defined similarly.  
We denote the set of all Player-$p$ strategies by
$\Strat_p$.  
The \emph{outcome} $\pat(\strat_1,\strat_2,s)$ of $\strat_1$ and
$\strat_2$ on $G$ starting at $s$ is the unique play
$\pat=\pat_0\pat_1\dots$ such that for all $i\ge 0$, if
$\pat_i\in{S_p}$, then $\pat_{i+1}=\strat_p(\pat_0\dots\pat_i)$ and
$\pat_0=s$.

A strategy 
 $\strat_\p \in\Strat_\p$  is memoryless if for any two sequences
$\sigma=s_0\dots s_i\in S^*S_\p$ and $\sigma'=s'_0\dots s'_{i'} \in
S^*S_\p$ such that $s_i=s'_{i'}$, we have
$\strat_\p(\sigma)=\strat_\p(\sigma')$.  We represent a memoryless
strategy $\strat_\p$ simply as a function from $S_\p$ to $S$.
A strategy is a \emph{finite-memory strategy} if it needs only finite
memory of the past, consisting of a
finite-state machine that keeps track of the history of the play.  The
strategy chooses a move depending on the state of the machine and the
location in the game.
Strategies that are not finite-memory are called
\emph{infinite-memory strategies}.

\paragraph{Quantitative and qualitative objectives.}  We
consider different objectives for the players.  A \emph{quantitative}
objective $f$ is a function $f: \Pat \to \real^\dim$ that assigns a
vector of reals as reward to every play.  We consider complementary
objectives for the two players; i.e., if the objective for Player~1 is
$f$, then the objective for Player~2 is $-f$.  The goal of each player
is to maximize her objective.  Note that Player~2 tries to minimize
$f$ by maximizing the complementary $-f$.
An objective $f:\Pat \to \set{0,\pm 1}$ that maps to the set $\set{0,
  1}$ (or $\set{0, -1}$) is per definition a \emph{qualitative
  objective}.  Given a qualitative objective $f:\Pat \to V$ we say a
play $\pat \in \Pat$ is \emph{winning for Player~1} if $f(\pat)=\max(V)$
holds, otherwise the play is \emph{winning for Player~2}.

\paragraph{Value.}  Given an objective $f$, the
\emph{Player-1 value of a state $s$ for a strategy $\strat_1$} is the
minimal value Player~1 achieves using $\strat_1$ against all Player-2
strategies, i.e.,
$
\Value_1(f, s,\strat_1) = \inf_{\strat_2 \in \Strat_2}
f(\pat(\strat_1,\strat_2,s)).
$
The \emph{Player-1 value of a state $s$} is the maximal value Player-1
can ensure from state $s$, i.e., 
$
\Value_1(f,s) = \sup_{\strat_1\in \Strat_1} \Value_1(f, s,\strat_1).
$
Player-2 values are defined analogously.  If $\Value_1(f, s) +
\Value_2(-f,s) = 0$ for all $s$, then the game is \emph{determined}
and we call $\Value_1(f, s)$ the \emph{value of $s$.}

\paragraph{Optimal, $\epsilon$-optimal, and winning
  strategies.}
Given an objective $f$ and a vector $\vec{\epsilon}\ge\vec{0}$, a
Player-1 strategy $\strat_1$ is \emph{Player-1
  $\vec{\epsilon}$-optimal from a state $s$} if 
$\Value_1(f, s,\strat_1) \ge \Value_1(f, s)-\vec{\epsilon}$.  If~$\strat_1$ is
$\vec{0}$-optimal from $s$, then we call $\strat_1$ 
\emph{optimal from $s$}.
Optimality for
Player-2 strategies is defined analogously.
If $f:\Pat \to V$ is a qualitative objective, a strategy $\strat_1$
is \emph{winning for Player~1} from $s$ if $ \Value_1(f, s,\strat_1) =
\max(V)$.


We now define various objectives.

\paragraph{Parity objectives.}
A \emph{parity objective} consists of a \emph{priority function}
$\priority: S \maps \set{0,1,\dots,k}$ that maps every state in $S$ to
a number (called \emph{priority}) between~$0$ and~$k$.  We denote by
$\maxpriority$ the maximal priority (i.e., $\maxpriority=k$).
The objective function $P$ of Player~1 maps a play $\pat$ to $1$ if
the smallest priority visited infinitely often is even, otherwise
$\pat$ is mapped to $0$.


\paragraph{Lexicographic mean-payoff objectives.}
A \emph{lexicographic mean-payoff objective} consists of a
\emph{reward function} $\lreward: E \maps \nat^\dim$ that maps every
edge in $G$ to a $\dim$-vector (called \emph{reward}) of natural
numbers.
We define  
$|\lreward|= \prod_{1 \leq i \leq \dim} \max_{e
  \in E} \reward_i(e)$, where $\reward_i(e)$ is the $i$-component of
$\lreward(e)$.  
The objective function of Player~1 for a play $\pat$ is the
lexicographic mean payoff
$\LM_\lreward(\pat) = \lim\inf_{n\maps\infty}
\frac{1}{n} \sum_{i=0}^{n} \lreward(\tuple{\pat_i,\pat_{i+1}}).
$
If $\dim=1$, then $\LM_\lreward(\rho)$ is the \emph{mean payoff}
\cite{Ehrenf79} and we refer to it as~$M_\lreward(\rho)$.

\paragraph{Lexicographic mean-payoff parity
  objectives.}  A \emph{lexicographic mean-payoff parity objective}
has a priority function $\priority: S \maps \set{0,1,\dots,k}$ and a
reward function $\lreward: E \maps \nat^\dim$.  The \emph{
  lexicographic mean-payoff parity value} $\LMP_\lreward(\pat)$ for
Player~1 of a play $\pat$ is the lexicographic mean-payoff
$\LM_{\lreward}(\pat)$ if $\pat$ is winning for the parity objective
(i.e., $P_\priority(\pat)=1$), else the payoff is $\minimum$.
If $\dim=1$, then $\LMP_{\lreward,\priority}(\pat)$ defines the
\emph{mean-payoff parity value} \cite{Chatte05} and we write
$\MP_{\lreward,\priority}(\pat)$.  If $\priority$ or $\lreward$ are
clear from the context, we omit them.

\paragraph{Games and automata.} A \emph{game} is a tuple
$\game=\tuple{G, f}$ consisting of a game graph $G=\tuple{S, \init, E}$
and an objective $f$.  An \emph{automaton} is a game with only one
player, i.e., $S = S_1$.  We name games and automata after their objectives.
\subsection{Lexicographic mean-payoff games} 
\label{sec:lex-mean-payoff-games}

In this section, we prove that memoryless strategies are sufficient for
lexicographic mean-payoff games, 
and present an algorithm to decide these games by a reduction to simple mean-payoff
games.
We first present the solution of lexicographic mean-payoff games with
a reward function with two components, and then extend it to
$\dim$-dimensional reward functions. Consider a lexicographic
mean-payoff game $\game_{\LM}=\tuple{\tuple{S, \init, E},\lreward}$,
where $\lreward=\tuple{r_1,r_2}$ consists of two reward functions.


\paragraph{Memoryless strategies suffice.}
We show that 
memoryless strategies suffice for both players by a reduction to a
finite cycle forming game.  Let us assume we have solved the
mean-payoff game with respect to the reward function $r_1$.  Consider
a \emph{value class} of $r_1$, i.e., a set of states having the same value
with respect to $r_1$.  It is not possible for Player~1 to move to a
higher value class, and Player~1 will never choose an edge to a lower
value class.  Similarly, Player~2 does not have edges to a lower value
class and will never choose edges to a higher value class.  Thus, we
can consider the sub-game for a value class.

Consider a value
class of value $\ell$ and
the sub-game induced by the value class.
We now play the following finite-cycle forming game: Player~1 and
Player~2 choose edges until a cycle $C$ is formed.  
The payoff for the game is as follows:
\begin{enumerate}
\item
 If the mean-payoff value of the cycle $C$ for $r_1$ is greater
  than $\ell$, then Player~1 gets reward $|r_2|+1$.
 \item
   If the mean-payoff value of the cycle $C$ for $r_1$ is smaller
  than $\ell$, then Player~1 gets reward~$-1$.
\item
  If the mean-payoff value of the cycle $C$ for $r_1$ is exactly
  $\ell$, then Player~1 gets the mean-payoff value for reward $r_2$ of
  the cycle $C$.
 \end{enumerate}




\begin{Lemma}\label{lemm-claim1}\label{lemm-claim2} 
  The value of Player~1 for any state in the finite-cycle forming game
  is $(i)$ strictly greater than~$-1$ and $(ii)$ strictly less
  than~$|r_2|+1$.
\end{Lemma}
\begin{proof} 
  Since all $r_2$ rewards are positive, a memoryless optimal strategy for 
  Player~1 for reward $r_1$ ensures~$(i)$.  Since all $r_2$ rewards are less
  than $|r_2|+1$, a Player~2 memoryless optimal strategy for reward
  $r_1$ ensures~$(ii)$. \qed
\end{proof}

\begin{Lemma}\label{lemm-claim3} 
  Both players have memoryless optimal strategy in the finite-cycle
  forming game.
\end{Lemma}
\begin{proof}
  The result can be obtained from the result of Bj\"orklund
  et al.\ \cite{Bjorkl04}.  From Theorem~5.1 and the comment in
  Section~7.2 it follows that in any finite-cycle forming game in which the
  outcome depends only on the vertices that appear in the cycle
  (modulo cyclic permutations) we have that memoryless optimal strategies
  exist for both players.  Our finite-cycle forming game satisfies the
  required conditions.  
  \qed
\end{proof}

\smallskip\noindent{\bf Example.} Before presenting the main lemma,
we discuss a subtle issue.
Consider two sequences $(a^1_i)_{i\geq 0}$ and $(a^2_i)_{i\geq 0}$ 
as follows: for all $i \geq 0$ we have (i)~$a^1_{2i}=1$ and $a^1_{2i+1}=2$; and 
(ii)~$a^2_{2i}=1$ and $a^1_{2i+1}=0$.
Then we have $\tuple{\liminf_{i \to \infty} a^1_i,\liminf_{i \to \infty} a^2_i}=
\tuple{1,0}$; however, under lexicographic ordering we have 
$\liminf_{i \to \infty} \tuple{a^1_i,a^2_i}=\tuple{1,1}$, since under lexicographic
ordering $\tuple{1,1}$ is smaller than $\tuple{2,0}$.
Note that in our semantics we consider liminf average of the reward vector under 
lexicographic ordering; and not the lexicographic ordering of components where
each component is the liminf average value.

\begin{Lemma} \label{lem:lmp-finitecycle}
  The following assertions hold.
\begin{enumerate}
\item If the value of the finite-cycle forming game is $\beta$ at a
  state $s$, then the value of the lexicographic mean-payoff game is
  $\tuple{\ell,\beta}$ at $s$.
\item A memoryless optimal strategy of the finite-cycle forming game is
  optimal for the lexicographic mean-payoff game.
\end{enumerate}
\end{Lemma}

\newcommand{\ov}{\overline}
\begin{proof}
The proof has the following two parts. 
\begin{enumerate}
\item We first show that a memoryless optimal strategy for 
Player~1 in the finite-cycle forming game ensures value 
at least $\tuple{\ell,\beta}$.
Fix a memoryless optimal strategy $\strat_1$ for Player~1 for
  the finite-cycle forming game: such a strategy exists by
  Lemma~\ref{lemm-claim3}.  Observe that by Lemma~\ref{lemm-claim1} we
  have $\beta>-1$.  In the resulting graph, for any cycle $C$
  reachable from $s$, the following assertions hold (because of 
  optimality of $\strat_1$): 
\begin{enumerate}
\item \emph{Property~1:} the mean-payoff reward for $r_2$ is at least $\beta$ 
and the mean-payoff reward for $r_1$ is at least $\ell$; or 

\item \emph{Property~2:} the mean-payoff reward for $r_1$ is at least 
$\ov{\ell} > \ell$ (i.e., the reward is strictly greater than $\ell$ and 
in this case the mean-payoff reward for $r_2$ can be less than $\beta$).

\end{enumerate} 
Consider any strategy $\strat_2$ for Player~2 and the path 
$\pat=\pat(s,\strat_1,\strat_2)$ and consider any prefix of length $n$ of 
$\pat$.  
The prefix can be decomposed as a finite-prefix of length at most $|S|$, 
then cycles in the graph, and then a trailing prefix of length at most $|S|$.
For a prefix of length $n$, let us denote by $J_1(n)$ the sum total of the 
steps of cycles in the prefix that satisfies property~1, and by $J_2(n)$ the
sum total of the steps of cycles in the prefix that satisfies property~2.
It follows that 
$\lim_{n\maps \infty} \frac{J_1(n) + J_2(n)}{n} \geq  
\lim_{n\maps \infty} \frac{n-2\cdot |S|}{n} =1$.
Since rewards are non-negative, it follows that for any $n$ we have 
\[
(a)~\sum_{i=0}^n r_1(\tuple{\pat_i,\pat_{i+1}}) 
\geq J_1(n) \cdot \ell + J_2(n) \cdot \ov{\ell};
\]
and
\[
(b)~\sum_{i=0}^n r_2(\tuple{\pat_i,\pat_{i+1}}) 
\geq J_1(n) \cdot \beta.
\]
Let $|R|=max\set{|r_1|,|r_2|}$.
Let $\lim\inf_{n \maps \infty} \frac{J_2(n)}{n}=\kappa$.
The following two case analysis completes the result.
\begin{enumerate}

\item If $\kappa>0$, then we have 

\[
\begin{array}{rcl}
\displaystyle
\lim\inf_{n\maps\infty} \frac{1}{n} \sum_{i=0}^{n}
  r_1(\tuple{\pat_i,\pat_{i+1}}) 
& \geq & 
\displaystyle
\lim\inf_{n\maps\infty} \left(\frac{J_1(n)}{n} \cdot \ell + \frac{J_2(n)}{n} \cdot \ov{\ell} \right) \\[2ex]
& = & 
\displaystyle
\lim\inf_{n\maps\infty} \left(\frac{J_1(n)}{n} \cdot \ell + \frac{J_2(n)}{n} \cdot (\ell+\ov{\ell}-\ell) \right) \\[2ex]
& \geq & 
\displaystyle
\lim\inf_{n\maps\infty} \left(\frac{J_1(n)+J_2(n)}{n} \cdot \ell\right) + 
\lim\inf_{n\maps\infty} \left(\frac{J_2(n)}{n} \cdot (\ov{\ell}-\ell) \right) \\[2ex]
& = & 
\displaystyle
\ell + \kappa \cdot (\ov{\ell}- \ell) > \ell.
\end{array}
\]
where the second inequality above comes from the standard fact that given two sequences $(a_n)_{n\geq 0}$ and
$(b_n)_{n \geq 0}$ of non-negative real numbers, we have 
$\lim\inf_{n \maps \infty}(a_n + b_n) \geq
\lim\inf_{n \maps \infty}(a_n) + \lim\inf_{n \maps \infty}(b_n)$;
and the final inequality uses that $\kappa>0$ and $\ov{\ell}> \ell$.
It follows that the lexicographic mean-payoff value
of $\pat$ is at least
$\tuple{\ell + \kappa \cdot (\ov{\ell}- \ell),0}
\geq \tuple{\ell,\beta}$.

\item If $\kappa=0$, then we have 
\[
\lim\inf_{n\maps\infty} \frac{1}{n} \sum_{i=0}^{n}
  r_1(\tuple{\pat_i,\pat_{i+1}}) 
\geq \ell.
\]
We now argue about the second component.
Fix $\epsilon>0$. 
Consider any length $n$ of the play prefix such that the average 
of the first component is at most $\ell+ \epsilon$.
Let $J(n)=J_1(n) + J_2(n)$, and note that $n-J(n) \leq 2\cdot |S|$.
Also note that $\ell \leq |R|$.
Thus we have the following inequalities:
\[
\begin{array}{rcl}
\ell + \epsilon & \geq &
\displaystyle 
\frac{1}{n} \sum_{i=0}^{n} r_1(\tuple{\pat_i,\pat_{i+1}}) \\[3ex]
& \geq & 
\displaystyle \frac{J_1(n)}{n}\cdot {\ell} + \frac{J_2(n)}{n}\cdot \ov{\ell}  \\[2ex]
& = & 
\displaystyle \frac{J_1(n)}{n}\cdot {\ell} + \frac{J_2(n)}{n}\cdot (\ov{\ell}-\ell +\ell)   \\[2ex]
& = & 
\displaystyle \frac{J_1(n)+J_2(n)}{n}\cdot {\ell} + \frac{J_2(n)}{n}\cdot (\ov{\ell}-\ell) \\[2ex]
& = & 
\displaystyle 
\frac{J(n)}{n}\cdot \ell +  
\frac{J_2(n)}{n}\cdot (\ov{\ell}-\ell)  \\[2ex]
& = & 
\displaystyle 
\ell + \frac{J_2(n)}{n}\cdot (\ov{\ell}-\ell) - \frac{n-J(n)}{n}\cdot \ell \\[2ex]
& \geq & 
\displaystyle 
\ell + \frac{J_2(n)}{n}\cdot (\ov{\ell}-\ell) - \frac{2\cdot |S|\cdot |R|}{n} 
\end{array}
\]
Hence it follows that $\frac{J_2(n)}{n}\leq \frac{\epsilon}{(\ov{\ell}-\ell)} + 
\frac{2\cdot |S|\cdot |R|}{n \cdot (\ov{\ell}-\ell)}$.
We now establish a bound on the average for the second component for the
length $n$.
We have
\[
\begin{array}{rcl}
\displaystyle 
\frac{1}{n} \sum_{i=0}^{n} r_2(\tuple{\pat_i,\pat_{i+1}}) 
& \geq & 
\displaystyle 
\frac{J_1(n)}{n}\cdot {\beta}
\\[2ex]
& = & 
\displaystyle 
\beta - \frac{J_2(n)}{n}\cdot \beta - \frac{(n-J(n))}{n}\cdot \beta \\[2ex]
& \geq & 
\displaystyle 
\beta - \frac{J_2(n)}{n}\cdot |R| - \frac{2\cdot |S|\cdot |R|}{n} \\[2ex]
& \geq & 
\displaystyle 
\beta -  \frac{\epsilon}{(\ov{\ell}-\ell)} \cdot |R| - 
\frac{2\cdot |S|\cdot |R|}{n\cdot(\ov{\ell}-\ell)} - \frac{2\cdot |S|\cdot |R|}{n} \\[3ex]
& = & 
\displaystyle 
\beta -  \frac{\epsilon}{(\ov{\ell}-\ell)} \cdot |R| - 
\frac{2\cdot |S|\cdot |R|}{n}\cdot\left(1+\frac{1}{(\ov{\ell}-\ell)}\right) 
\end{array}
\]
In the first equality above we write $\frac{J_1(n)}{n}= 1-\frac{J_2(n)}{n} - \frac{(n-J(n))}{n}$.
We also use that $\beta \leq |R|$ and 
$\frac{J_2(n)}{n}\leq \frac{\epsilon}{(\ov{\ell}-\ell)} + 
\frac{2\cdot |S|\cdot |R|}{n \cdot (\ov{\ell}-\ell)}$.
It follows that whenever the average of the first component is at most 
$\ell +\epsilon$, then the average of the second component is at least 
$\beta -  \frac{\epsilon}{(\ov{\ell}-\ell)} \cdot |R| - 
\frac{2\cdot |S|\cdot |R|}{n} \cdot\left( 1 + \frac{1}{(\ov{\ell}-\ell)}\right)$,
where $n$ is the length of the play.
Now fix $\epsilon_1>0$. 
Consider $n_0\in \nat$ such that 
$\frac{2\cdot |S|\cdot |R|}{n_0} \cdot\left( 1 + \frac{1}{(\ov{\ell}-\ell)}\right) \leq \epsilon_1$.
Consider any prefix of $\pat$ of length $n_1$ such that $n_1 \geq n_0$.
Then for the prefix either the average of the first component is 
at least $\ell+\epsilon$; or the average of the second component is at least
$\beta - \frac{\epsilon}{(\ov{\ell}-\ell)} \cdot |R| -\epsilon_1$.
We know the liminf average of the first component is $\ell$. 
It follows that the liminf average vector is at least 
$\tuple{\ell,\beta - \frac{\epsilon}{(\ov{\ell}-\ell)} \cdot |R|}$.
Since $\epsilon>0$ is arbitrary, it follows that 
the liminf average vector is at least 
$\tuple{\ell,\beta}$.
This completes the argument.

\end{enumerate}
It follows that $\LM(\pat) \geq \tuple{\ell,\beta}$.

\item Fix a memoryless optimal strategy $\strat_2$ for Player~2 for
  the finite-cycle forming game: such a strategy exists by
  Lemma~\ref{lemm-claim3}.
  Observe that by Lemma~\ref{lemm-claim2} we
  have $\beta<|r_2|+1$.  Then in the resulting graph, for any cycle $C$
  reachable from $s$, the following properties hold due to optimality of 
  $\strat_2$: 
 \begin{enumerate} 
 \item \emph{Property~1.} the mean-payoff reward for $r_1$ is at
  most $\ell$ and the mean-payoff reward for $r_2$ is at most $\beta$; or 
 \item \emph{Property~2.} the mean-payoff reward for $r_1$ is at most 
  $\ov{\ell} < \ell$ (the mean-payoff reward for $r_1$ is strictly smaller 
  than $\ell)$.
\end{enumerate}
By analysis similar to the previous case, it follows that the lexicographic
mean-payoff value is at most $\tuple{\ell,\beta}$. \qed
\end{enumerate}
\end{proof}


\paragraph{Reduction to mean-payoff games.}
\label{sec:reduction}
We now sketch a reduction of lexicographic mean-payoff games to
mean-payoff games for optimal strategies.
%
%
We reduce the reward function $\lreward=\tuple{r_1, r_2}$ to a single
reward function $r^*$.  We ensure that if the mean-payoff difference
of two cycles $C_1$ and $C_2$ for reward $r_1$ is positive, then the
difference in reward assigned by $r^*$ exceeds the largest possible
difference in the mean-payoff for reward $r_2$.  Consider two cycles
$C_1$ of length $n_1$ and $C_2$ of length $n_2$, such that the sum of
the $r_1$ rewards of $C_i$ is $\alpha_i$.  Since all rewards are
integral, $|\frac{\alpha_1}{n_1} - \frac{\alpha_2}{n_2}| >0$ implies
$| \frac{\alpha_1}{n_1} - \frac{\alpha_2}{n_2}| \geq \frac{1}{n_1
  \cdot n_2}$.
Hence we multiply the $r_1$ rewards by $m=|S|^2 \cdot |r_2|+1$ to
obtain $r^* = m \cdot r_1 + r_2$.  This ensures that if the
mean-payoff difference of two cycles $C_1$ and $C_2$ for reward $r_1$
is positive, then the difference exceeds the difference in the
mean-payoff for reward $r_2$.  Observe that we restrict our attention
to cycles only since we have already proven that optimal memoryless
strategies exist.

We can easily extend this reduction to reduce
lexicographic mean-payoff games with arbitrarily many reward functions 
to mean-payoff games.  
The following theorem follows from this reduction in combination with
known results for mean payoff parity games \cite{Ehrenf79,Zwick96}.  
\begin{theorem}[Lexicographic mean-payoff games]
  \label{thm:lex-mean-payoff-games}
  For all lexicographic mean-payoff games $\game_{\LM}=\tuple{\tuple{S, \init, E},\lreward}$, the following
  assertions hold: 
 \begin{enumerate} 
 \item {\em (Determinacy.) } For all states $s\in S$, we have 
  $\Value_1(\LM, s) + \Value_2(-\LM, s) = \vec{0}$.
 \item {\em (Memoryless optimal strategies.)} Both players have
   memoryless optimal strategies from every state $s\in S$.
 \item {\em (Complexity).} Whether the lexicographic mean-payoff value
   vector at a state $s \in S$ is at least a rational value
   vector~$\vec{v}$ can be decided in NP $\cap$ coNP.
 \item {\em (Algorithms).}  The lexicographic mean-payoff value vector
   for all states can be computed in time $O(|S|^{2\dim + 3}\cdot |E| \cdot
   |\lreward| )$.
\end{enumerate}
\end{theorem}

\subsection{Lexicographic Mean-Payoff Parity Games}
\label{sec:lex-mean-payoff-parity-games}

Lexicographic mean-payoff parity games are a natural lexicographic
extension of mean-payoff parity games~\cite{Chatte05}.
The algorithmic solution for mean-payoff parity games is a recursive
algorithm, where each recursive step requires the solution of a parity
objective and a mean-payoff objective.  The key correctness argument
of the algorithm relies on the existence of memoryless optimal
strategies for parity and mean-payoff objectives.  Since memoryless
optimal strategies exist for lexicographic mean-payoff games, the
solution of mean-payoff parity games extends to lexicographic
mean-payoff parity games: in each recursive step, we replace the
mean-payoff objective by a lexicographic mean-payoff objective.
Moreover, recent results of~\cite{CD12} that shows mean-payoff parity games
lie in NP $\cap$ coNP also extend to lexicographic mean-payoff parity games.
Thus, we have the following result.

\begin{theorem} [Lexicographic mean-payoff parity games]
  \label{thm:lex-mean-payoff-parity-games}
  For all lexicographic mean-payoff parity games
  $\game_{\LMP}=\tuple{\tuple{S, \init, E},\lreward,\priority}$,
  the following assertions hold. 
   \vspace{-5pt}
\begin{enumerate} 
 \item {\em (Determinacy).}  $\Value_1(\LMP,s) + \Value_2(-\LMP,s) =
   \vec{0}$ for all state $s \in S$.
   
 \item {\em (Optimal strategies).} 
   Optimal strategies for Player~1 exist but may require infinite memory; 
   memoryless optimal strategies exist for Player~2. 

 \item {\em (Complexity).} Whether the value at a state $s \in S$ is
   at least the vector~$\vec{v}$ can be decided in NP $\cap$ coNP.

 \item {\em (Algorithms).} The value for all states can be computed in
   time $O\big(|S|^\maxpriority\cdot(
   \min\set{|S|^{\frac{\maxpriority}{3}}\cdot |E|, |S|^{O(\sqrt{S})} }
   + |S|^{2\dim + 4} \cdot |E| \cdot |\lreward|)\big)$.
\end{enumerate}
\end{theorem}



In the following, we prove two properties of mean-payoff parity games
that are interesting for synthesis.  For simplicity, we present the
results for mean-payoff parity games.  The results extend to
lexicographic mean-payoff parity games as in 
Theorem~\ref{thm:lex-mean-payoff-parity-games}.
First, we show that the algorithm of \cite{Chatte05} can be adapted to
compute finite-memory strategies that are $\ve$-optimal.  Then, we show
that Player~1 has a finite-memory optimal strategy if and only if she
has a memoryless optimal strategy.

\paragraph{Finite-memory $\ve$-optimal strategy.} 
In mean-payoff parity games, though optimal strategies require
infinite memory for Player~1, there is a finite-memory
$\epsilon$-optimal strategy for every $\ve>0$.  The proof of this
claim is obtained by a more detailed analysis of the optimal strategy
construction of~\cite{Chatte05}.  The optimal strategy constructed
in~\cite{Chatte05} for Player~1 can be intuitively described as
follows.  The strategy is played in rounds, and each round has three
phases: (a) playing a memoryless optimal mean-payoff strategy; 
(b) playing a strategy in a sub-game; (c) playing a memoryless attractor 
strategy to reach a desired priority.  
Then the strategy proceeds to the next round.  
The length of the first phase is monotonically increasing in the number of
rounds, and it requires infinite memory to count the rounds.  Given an
$\ve>0$, we can fix a bound on the number of steps in the first phase
that ensures a payoff within $\ve$ of the optimal value.  Hence, a
finite-memory strategy can be obtained.  

\begin{figure}[tb]
  \centerline{\def\rechts{45}
\def\oben{20}
\def\links{0}
\def\unten{3}

\begin{picture}(\rechts,\oben)(\links,\unten)
{\fsize
\gasset{loopdiam=\loop,Nw=\Nw,Nh=\Nh,Nmr=\Nmr}

\node[Nmarks=i](q0)(10,10){$s_0$}
\node[Nmarks=nr](q1)(40,10){$s_1$}

\drawloop[ELside=l, ELdist=1](q0){$10$}
\drawedge[ELpos=50, ELside=l, ELdist=1, curvedepth=6](q0,q1){$10$}
\drawedge[ELpos=50, ELside=l, ELdist=1, curvedepth=6](q1,q0){$0$}
}
\end{picture}

}
  \caption{Game in which the optimal strategy requires infinite
    memory.}
  \label{fig:fin-memory}
\end{figure}
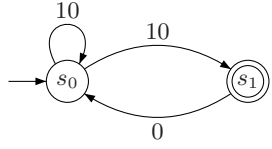
We illustrate the idea with
an example.  Consider the example shown in Fig.~\ref{fig:fin-memory}
where we have a game graph where all states belong to Player~1.  The
goal of Player~1 is to maximize the mean-payoff while ensuring that
state $s_1$ is visited infinitely often.  An optimal strategy is as
follows: the game starts in round~1.  In each round $i$, the edge
$s_0 \to s_0$ is chosen $i$ times, then the edge $s_0 \to s_1$ is
chosen once, and then the game proceeds to round $i+1$.  Any optimal
strategy in the game shown requires infinite memory.  However, given
$\ve>0$, in every round the edge $s_0 \to s_0$ can be chosen a fixed
number $K$ times such that $K > \frac{10}{\ve}-2$.
Then the payoff is 
$\frac{10\cdot K + 10}{K+2} =10 - \frac{10}{K+2} \geq 10 -\ve$ (since  $K > \frac{10}{\ve}-2$);
which is within $\ve$ of the value.  
It may also be noted that given $\ve>0$, the finite-memory optimal 
strategy can be obtained as follows.
We apply the recursive algorithm to solve the game to obtain  
two memoryless strategies: one for the mean-payoff strategy and 
other for the attractor strategy.
We then specify the bound (depending on $\ve$) on the number of 
steps for the mean-payoff strategy for each phase 
(this requires an additional $O(\frac{1}{\ve})$ time for the 
strategy description after the recursive algorithm).

\begin{theorem}\label{thm:lex-mean-payoff-parity-epsilon}
  For all lexicographic mean-payoff parity games and for all $\ve>0$,
  there exists a finite-memory $\ve$-optimal strategy for Player~1.
  Given $\ve>0$, a finite-memory $\ve$-optimal strategy can be constructed 
  in time $O(|S|^\maxpriority \cdot |E|^{2\dim + 6} \cdot |\lreward| + \frac{1}{\ve})$.
\end{theorem}

\paragraph{Optimal finite-memory and memoryless strategies.}
Consider a mean-payoff parity game
$\game=\tuple{\tuple{S, \init, E},\reward,\priority}$.  Our goal is to
show that if there is a finite-memory optimal strategy for Player~1,
then there is a memoryless optimal strategy for Player~1.  Suppose
there is a finite-memory optimal strategy $\wh{\pi}_1$ for Player~1.
Consider the finite graph $\wh{\game}$ obtained by fixing the strategy
$\wh{\pi}_1$. ($\wh{\game}$ is obtained as the synchronous product of
 the given game graph and finite-state strategy automaton for
 $\wh{\pi}_1$.)  
For a state $s \in S$, consider any cycle $\wh{C}$
in $\wh{\game}$ that is reachable from $\tuple{s,q_0}$ (where $q_0$
is the initial memory location) and $\wh{C}$ is executed to ensure
that Player~1 does not achieve a payoff greater than the value of the
game from~$s$.  We denote by $\wh{C}|_\game$ the sequence of states
in $\game$ that appear in $\wh{C}$.
We call a cycle $C$ of $\game$ that appears in $\wh{C}|_\game$ a
\emph{component cycle of $\wh{C}$}.  We have the following properties
about the cycle $\wh{C}$ and its component cycles. 
  \vspace{-5pt}
\begin{enumerate}
\item $\min(\priority(\wh{C}|_\game))$ is even.
\item 
  Suppose there is a component cycle $C$ of $\wh{C}$ such that
  the average of the rewards of $C$ is greater than the average of the
  rewards of $\wh{C}$. If Player~2 fixes a finite-memory strategy that corresponds to 
  the execution of cycle $\wh{C}$,  then an infinite-memory strategy can be
  played by Player~1 that pumps the cycle~$C$ longer and longer to
  ensure a payoff that is equal to the average of the weights of $C$.
  The infinite memory strategy ensures that all states in
  $\wh{C}|_\game$ are visited infinitely often, but the long-run
  average of the rewards is the average of the rewards of $C$.  This
  would imply that for the cycle $\wh{C}$, Player~1 can switch to an
  infinite-memory strategy and ensure a better payoff.
\item If there is component cycle $C$ of $\wh{C}$ such that
  $\min(\priority(C)) > \min(\priority(\wh{C}|_\game))$, then the cycle
  segment of $C$ can be ignored from $\wh{C}$ without affecting the
  payoff.
\item Suppose we have two component cycles $C_1$ and $C_2$ in
  $\wh{C}$ such that $\min(\priority(C_1)) = \min(\priority(C_2)) =
  \min(\priority(\wh{C}|_\game))$, then one of the cycles can be
  ignored without affecting the payoff.
\end{enumerate}
It follows from above that if the finite-memory strategy $\wh{\pi}_1$
is an optimal one, then it can be reduced to a strategy $\pi_1'$ such
that if Player~2 fixes a finite-memory counter-optimal strategy $\pi_2$, 
then every cycle $C$ in the finite graph obtained from fixing
$\pi_1'$ and $\pi_2$ is also a cycle in the original game graph.  
Since finite-memory optimal strategies exist for Player~2, 
a correspondence of the value of the game and the value of the following 
finite-cycle forming game can be established.  
The finite-cycle forming game is played on $\game$ and the game stops when a 
cycle~$C$ is formed, and the payoff is as
follows: if $\min(\priority(C))$ is even, then the payoff for Player~1
is the average of the weights of the $C$, otherwise the payoff for
Player~1 is $\minimum$. 
The existence of pure memoryless optimal
strategy in the finite-cycle forming game can be obtained from the
results of Bj\"{o}rklund et al.~\cite{Bjorkl04}.  This concludes the proof
of the following theorem.

\begin{theorem}
  \label{thm:optimal}
  For all lexicographic mean-payoff parity games, if Player~1 has a
  finite-memory optimal strategy, then she has a memoryless optimal
  strategy. 
\end{theorem}

It follows from Theorem~\ref{thm:optimal} that the decision whether
there is a finite-memory optimal strategy for Player~1 is in NP.  The
NP procedure goes as follows: we guess the value $\vinit$ of state $\init$
and verify that the value at $\init$ is no more than $\vinit$.  We can
decide in coNP whether the value at a state is at least $v$, for $v
\in \mathbb{Q}$.  Thus, we can decide in NP whether the value at state
$\init$ is no more than $\vinit$ (as it is the complement).  Then, we guess
a memoryless optimal strategy for Player~1 and verify (in polynomial
time) that the value is at least~$\vinit$ given the strategy.


\section{Quantitative Verification and Synthesis}
\label{sec:verification-synthesis}

We are interested in the verification and the synthesis problem for
quantitative specifications given by a lexicographic mean-payoff
(parity) automaton.  In the following simple lemma we establish
that these automata also suffice to express qualitative properties.

\begin{lemma}
  Let $A = \tuple{G,\priority}$ be a deterministic parity automaton
  and let $A' = \tuple{G', \lreward}$ be a lexicographic mean-payoff
  automaton. We can construct a lexicographic mean-payoff parity automaton $A
  \times A' = \tuple{G \times G', \lreward, \priority}$, where $G
  \times G'$ is the product graph of $G$ and $G'$ 
  such that for any word $w$ and associated run $\pat$, $\LMP_{A \times
    A'}(\pat) = \minimum$ if the run of $w$ is lost in $A$, and
  $\LM_{A'}(\pat')$ otherwise, where $\pat'$ is the projection of
  $\pat$ on $G'$. 
\end{lemma}
Note that $\pair{\lang{A},\lang{A'}} = \lang{A \times A'} + \vec{1}$,
assuming that $\inf_{w\in\A^\omega} L_{A'}(w) =0$.  If $A$ is a safety automaton, the
language $\pair{\lang{A},\lang{A'}}$ can be presented by a
lexicographic mean-payoff automaton (see
Example~\ref{ex:game-strategy-mealy}).
Thus, lexicographic mean-payoff automata
suffice to express both a quantitative aspect and a safety aspect of a
specification.  Lexicographic mean-payoff parity automata can be used
to introduce a quantitative aspect to liveness specifications and thus
to the usual linear-time temporal logics.

\begin{Example}\label{ex:game-strategy-mealy}
  Let us resume Example~\ref{ex:intro}.  Fig.~\ref{fig:lemma1} shows a
  safety automaton $B$ for the specification $\always(r \rightarrow \g
  \vee \nextt \g)$.  It also shows the mean-payoff automaton $C$ for
  $\pair{\lang{B},\lang{A_2}}$.  (See Fig.~\ref{fig:unnecessary} for
  the definition of $A_2$.)
    \begin{figure}[tb]
      \centering
       \def\rechts{115}
\def\oben{35}
\def\links{-10}
\def\unten{3}

\begin{picture}(\rechts,\oben)(\links,\unten)
{\fsize
\put(-10,20){$B$}
\gasset{loopdiam=\loop,Nw=\Nw,Nh=\Nh,Nmr=\Nmr}

\node[Nmarks=i](q0)(0,15){$q'_0$}
\node[Nmarks=n](q1)(20,15){$q'_1$}
\node[Nmarks=n](q2)(40,15){$q'_2$}

\drawloop[ELside=l, ELdist=1](q0){$\begin{array}{c}\nr(1)\\ \g(1)\end{array}$}
\drawloop[ELside=l, ELdist=1](q2){$\top(0)$}

\drawedge[ELpos=50, ELside=l, ELdist=1, curvedepth=6](q0,q1){$\r\ng(1)$}
\drawedge[ELpos=50, ELside=l, ELdist=1, curvedepth=6](q1,q0){$\g(1)$}
\drawedge(q1,q2){$\ng(0)$}

\put(60,20){$C$}
\node[Nmarks=i](q10)(70,15){}
\node[Nmarks=n](q11)(90,15){}
\node[Nmarks=n](q12)(110,15){}

\drawloop[ELside=l, ELdist=1](q10){
  $\begin{array}{c}\nr\ng(2)\\ \r\g(2)\end{array}$}
\drawloop[ELside=l, ELdist=1, loopangle=270](q10){$\nr\g(1)$}
\drawloop[ELside=l, ELdist=1](q12){$\top(0)$}

\drawedge[ELpos=50, ELside=l, ELdist=1, curvedepth=6](q10,q11){$\r\ng(2)$}
\drawedge[ELpos=50, ELside=l, ELdist=1, curvedepth=6](q11,q10){$\g(2)$}
\drawedge(q11,q12){$\ng(0)$}

}

\end{picture}

      \caption{Safety automaton $B$ for  $\always(\req \implies
        \grant\vee \nextt\grant)$ and automaton $C$ for $\pair{\lang{B},\lang{A_2}}$.}
    \label{fig:lemma1}
  \end{figure}
\end{Example}

\subsection{Quantitative Verification}
\label{sec:verification}

We now consider the verification problem for quantitative
specifications.  For qualitative specifications, the verification
problem is whether an implementation satisfies the specification for
all inputs.  For quantitative specifications, the problem generalizes
to the question if an implementation can achieve a given value
independent of the inputs.

Let $A = \tuple{\tuple{S, \init, E}, \lreward, \priority}$ be a
lexicographic mean-payoff parity automaton and let $M = \tuple{Q, q_0,
  \delta}$ be a Mealy machine.  The \emph{quantitative verification
  problem} is to determine $\lang{A}(M)$.  The
corresponding decision problem is whether $\lang{A}(M) \geq c$ for a
given cutoff value $c$.  Clearly, verification of qualitative
languages is a special case in which the cutoff value is $1$.

\begin{theorem}
  \label{thm:verify-ms}
  The value $\lang{A}(M)$ can be computed in time
  $O(|S| \cdot |Q| \cdot |E| \cdot |\delta| \cdot d \cdot \lg(|Q| \cdot
  |\delta| \cdot |\lreward|) )$.
\end{theorem}
\begin{proof}
  We reduce the lexicographic mean-payoff parity automata to a
  mean-payoff parity automaton $A'$ using the reduction stated in
  Section~\ref{sec:lex-mean-payoff-games} and build the product
  automaton of $A'$ and $M$.  Then, we check if it contains a cycle
  that is not accepted by the parity algorithm \cite{King01}.  If so,
  we return $\minimum$.  If not, in the second step we find the
  minimal mean-weight cycle \cite{Karp78}.
\end{proof}

\begin{Example}
  In Example~\ref{ex:intro}, we computed the values of
  Implementations~$\M_1$, $\M_2$, and $\M_3$ (Fig.~\ref{fig:systems})
  for the specifications $A_1$ and $A_2$ given in
  Fig.~\ref{fig:unnecessary}.  Specification $A_1$ requires the number
  of grants to be minimal.  Under this specification, $\M_3$ is
  preferable to both other implementations because it only produces
  half as much grants in the worst case.  Unfortunately, $A_1$ treats
  a grant the same way regardless of whether a request occurred.
  Thus, this specification does not distinguish between $\M_1$
  and~$\M_2$.  Specification $A_2$ only punishes ``unnecessary''
  grants, which means that $A_2$ prefers $\M_2$ and $\M_3$ to $\M_1$.

  A preference between the eagerness of $\M_2$ and the laziness of
  $\M_3$ can be resolved in either direction.  For instance, if we
  combine the two quantitative languages using addition, lazy
  implementations are preferred.
\end{Example}

\subsection{Quantitative Synthesis}
\label{sec:synthesis}
In this section, we show how to automatically construct an
implementation from a quantitative specification given by a
lexicographic mean-payoff (parity) automaton.  First, we show the
connection between automata and games, and between strategies and
Mealy machines, so that we can use the theory from
Sections~\ref{sec:games} to perform synthesis.  Then, we define
different notions of synthesis and give their complexity bounds.

We will show the polynomial conversions of an automata to a game and
of a strategy to a Mealy machines using an example.
\begin{Example}
  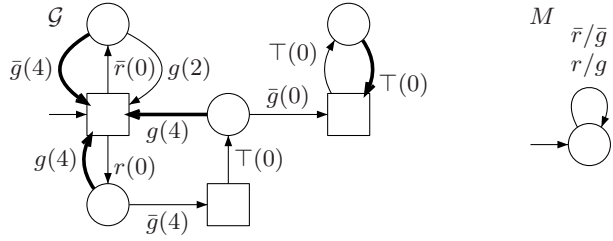
\begin{figure}[tbp]
    \centering
  \def\rechts{95}
\def\oben{35}
\def\links{-6}
\def\unten{0}

\begin{picture}(\rechts,\oben)(\links,\unten)
{\fsize
  \gasset{loopdiam=\loop,Nw=\Nw,Nh=\Nh,Nmr=\Nmr}
  
  \put(0,30){$\mathcal{G}$}
  \node[Nmarks=i,Nmr=0](q0)(10,15){}
  \node(q1)(10,0){}
  \node(q2)(10,30){}
  \node[Nmr=0](q3)(30,0){}
  \node(q4)(30,15){}
  \node(q5)(50,30){}
  \node[Nmr=0](q6)(50,15){}

  \drawedge[ELdist=1,ELpos=60](q0,q1){$\r (0)$}
  \drawedge[ELdist=1,curvedepth=4,linewidth=.6,AHlength=20,AHlength=0](q1,q0){$\g (4)$}
  \drawedge[ELdist=1,ELside=r](q0,q2){$\nr (0)$}
  \drawedge[ELdist=1,curvedepth=9](q2,q0){$\g (2)$}
  \drawedge[ELdist=1,curvedepth=-8,ELside=r,linewidth=.6,AHlength=20,AHlength=0](q2,q0){$\ng (4)$}
  \drawedge[ELdist=1,ELside=r](q1,q3){$\ng (4)$}
  \drawedge[ELdist=1,ELside=r](q3,q4){$\top (0)$}
  \drawedge[ELdist=1,linewidth=.6,AHlength=20,AHlength=0](q4,q0){$\g (4)$}
  \drawedge[ELdist=1](q4,q6){$\ng (0)$}
  \drawedge[ELdist=1,curvedepth=4,ELpos=60](q6,q5){$\top (0)$}
  \drawedge[ELdist=1,curvedepth=4,ELpos=60,linewidth=.6,AHlength=20,AHlength=0](q5,q6){$\top (0)$}

    \put(80,30){$M$}
  \node[Nmarks=i](qM)(90,10){}
  \drawloop(qM){$\begin{array}{c}\nr/\ng\\ r/g\end{array}$}
}  
\end{picture}

    \caption{A game (optimal strategy shown in bold) and corresponding Mealy machine}
    \label{fig:game-strategy-mealy}
  \end{figure}
  
  Fig.~\ref{fig:game-strategy-mealy}(left) shows the game
  $\mathcal{G}$ corresponding to the automaton~$C$ shown in
  Fig.~\ref{fig:lemma1}.  Note: The alphabet $2^{\AP}$ has been split
  into an input alphabet~$2^{\I}$ controlled by Player~2 (squares) and
  an output alphabet $2^{\O}$ controlled by Player~1 (circles).
  Accordingly, each edge $e$ of $C$ is split into two edges $e_2$ and
  $e_1$; the reward of $e_2$ is zero and the reward of $e_2$ is double
  the reward of $e$.  It should be clear that with the appropriate
  mapping between runs, the payoff remains the same.  Because we want
  a Mealy machine, the input player makes the first move.

  The figure also shows an optimal strategy (bold edges) for
  $\mathcal{G}$ with payoff $2$.  The right side of the figure shows
  the Mealy machine $M$ corresponding to the strategy.  It is
  constructed by a straightforward collection of inputs and chosen
  outputs.  It is easily verified that $\quan_{C}(M) = 2$.
\end{Example}

\begin{definition}
  Let $\quan$ be a quantitative language and let $\vc \in \real^d$ be
  a cutoff value.  We say that $\quan$ is \emph{$\vc$-realizable} if
  there is a Mealy machine $M$ such that $\quan(M) \geq \vc$.  We say
  that $\quan$ is \emph{limit-$\vc$-realizable} if for all
  $\vec{\epsilon}>0$ there is a Mealy machine $M$ such that $\quan(M) +
  \vec{\epsilon} \geq \vc$.

  Suppose the supremum of $\quan(M)$ over all Mealy machines $M$ exists,
  and denote it by $\vc^*$.  We call $\quan$ \emph{realizable
  (limit-realizable)} if $\quan$ is $\vc^*$-realizable
  (limit-$\vc^*$-realizable).  A Mealy machine~$M$ with value $L(M) \ge
  \vc^*$ ($\quan(M) + \vec{\epsilon} \geq \vc^*$) is called
  \emph{optimal} (\emph{$\vec{\epsilon}$-optimal}, resp.).

\end{definition}

\vspace{-1mm} Clearly, realizability implies limit-realizability.
Note that by the definition of supremum, $\quan$ is
limit-$\vc^*$-realizable iff $\vc^*$ is defined.  Note also that
realizability for qualitative languages corresponds to realizability
with cutoff 1.  Synthesis is the process of constructing an optimal
($\vec{\epsilon}$-optimal) Mealy machine.
Note that for a cutoff value $\vc$, if $L$ is $\vc$-realizable, then
we have that $L(M)\ge\vc$ for any optimal Mealy machine $M$.  If $L$
is limit-$\vc$-realizable, then $L(M_{\epsilon})+\vec{\epsilon}\ge\vc$
holds for any $\vec{\epsilon}$-optimal Mealy machine $M_{\epsilon}$.




\begin{Example}
  We have already seen an example of a realizable specification
  expressed as a mean-payoff automaton (See
  Figs.~\ref{fig:unnecessary} and~\ref{fig:lemma1} and
  Example~\ref{ex:game-strategy-mealy}.)
  Example~\ref{ex:limitrealizable} shows a language that is only
  limit-realizable.
\end{Example}



\noindent 
For the combination of safety and quantitative specifications, we
have Theorem~\ref{thm:safety}.
\begin{theorem}
  \label{thm:safety}
  Let $A = \tuple{\tuple{S, \init, E}, \lreward}$ be a lexicographic
  mean-payoff automaton of dimension $d$, and let $\vc$ be a cutoff
  value.  The following assertions hold.
  \vspace{-5pt}
  \begin{enumerate}
  \item $\quan_A$ is realizable (hence limit-realizable); 
  $\quan_A$ is $\vc$-realizable iff $\quan_A$ is limit-$\vc$-realizable. 
  \item $\vc$-realizability (and by~(1) limit-$\vc$-realizability) of
    $\quan_A$ are decidable in NP $\cap$ coNP. 
  \item An optimal Mealy machine can be constructed in time
    $O(|E|^{4d+6}\cdot|\lreward|)$. 
  \end{enumerate}
\end{theorem}
The first results follow from the existence of memoryless
optimal strategies for lexicographic mean-payoff games.  The second and 
third results follows from the complexity and algorithms of solving these games.  
(See Theorem~\ref{thm:lex-mean-payoff-games}.)  
For liveness, we have the following result.
\begin{theorem}
  Let $A = \tuple{\tuple{S, \init, E}, \lreward, \priority}$ be a
  lexicographic mean-payoff parity automaton of dimension $d$ and let
  $\vc$ be a cutoff value.  The following
  assertions hold.
    \vspace{-5pt}
  \begin{enumerate}
  \item $\quan_A$ is limit-realizable, but it may not be realizable; 
  limit-$\vc$-realizability of $\quan_A$ does not imply $\vc$-realizability.
  \item Realizability and $\vc$-realizability of $\quan_A$ are decidable in NP, and
  limit-$\vc$-realizability of $\quan_A$ is decidable in coNP.
  \item For $\vec{\epsilon}>0$, an $\vec{\epsilon}$-optimal Mealy machine can be constructed in time
  $O(|S|^{|p|} \cdot |E|^{4d +6} \cdot |\lreward| + \frac{1}{\epsilon})$. 
  If $\quan_A$ is realizable, then an optimal Mealy machine can be
  constructed in time $O(|S|^{|p|} \cdot |E|^{4d +6} \cdot |\lreward|)$. 
   \end{enumerate}
\end{theorem}
Explanation: Following Theorem~\ref{thm:optimal}, realizability and
$\vc$-realizability can be computed in NP.
We have that $\quan_A$ is limit-$\vc$-realizable iff $\vc$ is not
higher than the value of the initial state, which can be decided in
coNP.  (Theorem~\ref{thm:lex-mean-payoff-parity-games}.)
Limit-realizability follows from
Theorem~\ref{thm:lex-mean-payoff-parity-epsilon}.

\begin{example}
  In Example~\ref{ex:limitrealizable} we discussed the specification
  $\phi = \always(\r \rightarrow \eventually \g)$.  In combination
  with the quantitative language given by $A_3$ in
  Fig.~\ref{fig:distance}, this specification is optimally realizable
  by a finite implementation: implementations $M_1$ and $M_2$ from
  Fig.~\ref{fig:systems} are two examples.  The combination of $\phi$
  and the quantitative language given by $A_1$ in
  Fig.~\ref{fig:unnecessary} only yields a specification that is
  optimally limit-realizable.  Automaton $A_1$ prefers as few as
  possible requests.  An implementation that is optimal within $1/k$
  could simply give a request every $k$ cycles.  
  It may not be useful to require that something
  happens as infrequently as possible in the context of liveness
  specifications.  Instead, more subtle approaches are necessary; in this case
  we could require that \emph{unnecessary} grants occur as little as
  possible. (Cf.\ $A_2$ in Fig.~\ref{fig:unnecessary}.)
\end{example}

\section{Conclusions and Future Work}
\label{sec:conclusions}

We introduced a measure for the ``goodness'' of an implementation by
adding quantitative objectives to a qualitative specification.  Our
quantitative objectives are mean-payoff objectives, which
are combined lexicographically.  
Mean-payoff objectives are relatively standard and, as we
demonstrated, sufficiently expressive for our purposes.
Other choices, such as discounted
objectives~\cite{DiscountingTheFuture}, are possible as well.
These give rise to different expressive powers for specification 
languages~\cite{Chatte08a}.

Finally, we have taken the worst-case view that the quantitative value
of an implementation is the worst reward of all runs that the
implementation may produce.  There are several alternatives.  For
instance, one could take the average-case view of assigning to an
implementation some expected value of the cost taken over all possible
runs, perhaps relative to a given input distribution.  Another option
may be to compute admissible strategies.  It can be shown that such
strategies do not exist for all mean-payoff games, but they may exist
for an interesting subset of these games.

\smallskip\noindent{\bf Acknowledgements.} 
We thank V{\'e}ronique Bruy{\`e}re, No{\'e}mie Meunier, and Jean-Fran\c{c}ois Raskin
for insightful comments related related to complexity of strategies in 
lexicographic mean-payoff games.
We also thank  No{\'e}mie Meunier for suggestions of simplification of some
calculations of Lemma~3 using that rewards are non-negative.

\bibliographystyle{plain}
\bibliography{main} 

\end{document}